%% file: mainpaper.tex
\documentclass[11pt,a4paper]{article}

\usepackage{amsmath,amsfonts,amssymb,amsthm}
\usepackage{amsthm}
\usepackage{graphicx,color}
\usepackage{boxedminipage}
\usepackage[ruled,boxed,linesnumbered]{algorithm2e}
\usepackage{framed}
\usepackage{enumitem}
\usepackage{thmtools}
\usepackage{thm-restate}
\usepackage{xspace}
\usepackage{bm}
\usepackage{todonotes}
\usepackage{footnote}
\usepackage{mathtools}
\usepackage{mathrsfs}
\usepackage{vmargin}
\setmarginsrb{2.3cm}{2.3cm}{2.3cm}{2.3cm}{0pt}{0pt}{0pt}{6mm}
\usepackage{todonotes}
\usepackage[most]{tcolorbox}
\usepackage{footnote}
\usepackage{rotating}
 \usepackage[pdftex, plainpages = false, pdfpagelabels, 
                 bookmarks=true,
                 bookmarksopen = true,
                 bookmarksnumbered = true,
                 breaklinks = true,
                 linktocpage,
                 pagebackref,
                 colorlinks = true,  
                 linkcolor = blue,
                 urlcolor  = blue,
                 citecolor = red,
                 anchorcolor = green,
                 hyperindex = true,
                 hyperfigures
                 ]{hyperref} 
 \usepackage{xifthen}
 \usepackage{tabularx}
\usepackage{tikz} 
 \usetikzlibrary{calc}
 \usepackage{thm-autoref}
\usepackage[nameinlink]{cleveref}

\usepackage{soul}

\newtheorem{theorem}{Theorem}
\newtheorem{lemma}{Lemma}
\newtheorem{claim}{Claim}
\newtheorem{corollary}{Corollary}
\newtheorem{definition}{Definition}
\newtheorem{observation}{Observation}
\newtheorem{proposition}{Proposition}

\theoremstyle{definition}

\newlength{\RoundedBoxWidth}
\newsavebox{\GrayRoundedBox}
\newenvironment{GrayBox}[1]%
   {\setlength{\RoundedBoxWidth}{.93\textwidth}
    \def\boxheading{#1}
    \begin{lrbox}{\GrayRoundedBox}
       \begin{minipage}{\RoundedBoxWidth}}%
   {   \end{minipage}
    \end{lrbox}
    \begin{center}
    \begin{tikzpicture}%
       \node(Text)[draw=black!20,fill=white,rounded corners,%
             inner sep=2ex,text width=\RoundedBoxWidth]%
             {\usebox{\GrayRoundedBox}};
        \coordinate(x) at (current bounding box.north west);
        \node [draw=white,rectangle,inner sep=3pt,anchor=north west,fill=white] 
        at ($(x)+(6pt,.75em)$) {\boxheading};
    \end{tikzpicture}
    \end{center}}

\newenvironment{defproblemx}[2][]{\noindent\ignorespaces%
                                \FrameSep=6pt%
                                \parindent=0pt%
                \vspace*{-1.5em}
                \ifthenelse{\isempty{#1}}{%
                  \begin{GrayBox}{\textsc{#2}}%
                }{%
                  \begin{GrayBox}{\textsc{#2} parameterized by~{#1}}%
                }
                \begin{tabular*}{\textwidth}{@{\hspace{.1em}} >{\itshape} p{1.8cm} p{0.8\textwidth} @{}}%
            }{
                \end{tabular*}%
                \end{GrayBox}%
                \ignorespacesafterend
            }  

\newenvironment{defproblemxb}[2][]{\noindent\ignorespaces
  \FrameSep=6pt%
  \parindent=0pt%
  \vspace*{-1.5em}
  \ifthenelse{\isempty{#1}}{%
    \begin{GrayBox}{\textsc{#2}}%
    }{%
      \begin{GrayBox}{\textsc{#2} parameterized by~{#1}}%
      }
      \begin{tabular*}{\textwidth}{@{\hspace{.1em}} >{\itshape} p{1.2cm} p{0.85\textwidth} @{}}%
      }{
      \end{tabular*}%
    \end{GrayBox}%
    \ignorespacesafterend
  }

\newcommand{\pname}{\textsc}
\newcommand{\ProblemFormat}[1]{\pname{#1}}
\newcommand{\ProblemIndex}[1]{\index{problem!\ProblemFormat{#1}}}
\newcommand{\ProblemName}[1]{\ProblemFormat{#1}\ProblemIndex{#1}{}\xspace}
\newcommand{\probpropfair}{\ProblemName{Proportionally Fair Matching}}  
\newcommand{\lr}[1]{\left(#1\right)}
\newcommand{\LR}[1]{\left\{#1\right\}}
\newcommand{\real}{\mathbb{R}}

\newcommand{\np}{\textsf{NP}}
\newcommand{\prob}{\mathsf{Pr}}

\begin{document}

\title{Proportionally Fair Matching with Multiple Groups}

\author{Sayan Bandyapadhyay \thanks{Portland State University, USA}
	\and
	Fedor V. Fomin\thanks{University of Bergen, Norway}.
\and
Tanmay Inamdar\thanks{University of Bergen, Norway}
\and
Kirill Simonov\thanks{Technische Universität Wien, Austria}
}

\date{}

\maketitle

\begin{abstract}
  The study of fair algorithms has become mainstream in machine learning and artificial intelligence due to its increasing demand in dealing with biases and discrimination. Along this line, researchers have considered  fair versions of traditional optimization problems including clustering, regression, ranking and voting. However, most of the efforts have been channeled into designing heuristic algorithms,  which often do not provide any guarantees on the quality of the solution. In this work, we study matching problems with the notion of proportional fairness. Proportional fairness is one of the most popular notions of group fairness where every group is represented up to an extent proportional to the final selection size. Matching with proportional fairness or more commonly, proportionally fair matching, was introduced in  [Chierichetti et al., AISTATS, 2019], where the problem was studied with only two groups. However, in many practical applications, the number of groups---although often a small constant---is larger than two. In this work, we make the first step towards understanding the computational complexity of proportionally fair matching with more than two groups. We design exact and approximation algorithms achieving reasonable guarantees on the quality of the matching as well as on the time complexity. Our algorithms are also supported by suitable hardness bounds. 
\end{abstract}

\input{intro}

\input{prelims}

\input{proportionally}

\input{conclusion}

\bibliographystyle{siam}
\bibliography{ref}

\end{document}

%% file: intro.tex
\section{Introduction}
Machine learning (ML) algorithms are ubiquitous in today's world,  constantly playing crucial roles in decision-making which has an immeasurable impact on human lives. These algorithms trained on past instances are extremely powerful and most of the time output correct solutions without making any error. However, in recent times, these algorithms have faced critiques for being biased towards underrepresented groups \cite{angwin2019machine,datta2015automated,garb1997race,BuolamwiniG18}. Consequently, researchers have made efforts in understanding how biases are introduced in the ML pipeline and whether it is possible to get rid of them. This research has given rise to an entire subfield called \emph{fairness} in ML. All the work done so far in this budding subfield can broadly be classified into two types. The first one studies different notions of fairness and their interactions \cite{calders2010three,dwork2012fairness,feldman2015certifying,crowson2016assessing,kleinberg2017inherent,DBLP:conf/kdd/Corbett-DaviesP17}. These works essentially show that there is no universal definition of fairness that captures all the scenarios and it is not possible to satisfy different fairness notions simultaneously. In the second type of works, researchers have studied fair versions of classical problems incorporating suitable notions of fairness from the first type. Notably the problems considered include clustering \cite{chierichetti2017fair,huang2019coresets,bera2019fair,DBLP:journals/corr/abs-2007-10137}, regression \cite{DBLP:conf/nips/JosephKMR16,DBLP:journals/corr/BerkHJJKMNR17,agarwal2019fair}, ranking \cite{DBLP:conf/icalp/CelisSV18}, voting \cite{DBLP:conf/ijcai/CelisHV18} and matching \cite{DBLP:conf/aistats/Chierichetti0LV19}. 

In this paper, we consider the proportionally fair matching problem. Matching appears naturally in several applications in ML, e.g., assigning products to customers \cite{ristoski2018machine}; students to schools \cite{kurata2017controlled}; reviewers to manuscripts \cite{charlin2013toronto}; and workers to firms \cite{DBLP:conf/ijcai/AhmadiADFK20}. There are scores of works that study fair versions of matchings \cite{DBLP:conf/aistats/Chierichetti0LV19,DBLP:journals/algorithmica/HuangKM016,DBLP:journals/datamine/Garcia-SorianoB20,kamada2020fair,klaus2006procedurally,freemantwo}. Among these distinct notions of matchings, our work is most relevant to $(\alpha,\beta)$-balanced matching \cite{DBLP:conf/aistats/Chierichetti0LV19}. $(\alpha,\beta)$-balanced matching was formulated by \cite{DBLP:conf/aistats/Chierichetti0LV19} by bringing \emph{proportional} fairness and maximum cardinality matching together. Proportional fairness is based on the concept of disparate impact \cite{feldman2015certifying}, which in the context of matching is defined as follows. A matching is $(\alpha,\beta)$-balanced or proportionally fair if the ratio between the number of edges from each group and the size of the matching is at least $\alpha$ and at most $\beta$. 

As a motivating example of proportionally fair matching, 
consider the product recommendation problem in e-commerce. With the advancement of digital marketing and advertising, nowadays companies are interested in more fine-tuned approaches that help them reach the target groups of customers. These groups may be representative of certain underlying demographic categorizations into based on gender, age group, geographic location etc. Thus, the number of groups is often a small constant. In particular, in this contemporary setting, one is interested in finding assignments  that involve customers from all target groups and have a balanced impact on all these groups. This assignment problem can be modeled as the proportionally fair matching problem between customers and products. In a realistic situation, one might need to assign many products to a customer and many customers to a product. This can be achieved by computing multiple matchings in an iterative manner while removing the edges from the input graph that are already matched.  

In a seminal work, \cite{DBLP:conf/aistats/Chierichetti0LV19} obtained a polynomial-time $3/2$-approximation when the number of groups is 2.
However, in many real-world situations, like in the above example, it is natural to assume that the number of target groups is more than 2. Unfortunately, the algorithm of \cite{DBLP:conf/aistats/Chierichetti0LV19} strongly exploits the fact that the number of groups $\ell=2$.  
It is not clear how to adapt or extend their algorithm when we have more than two groups. 
The only known algorithm prior to our work for $\ell>2$ groups was an $n^{O(\ell)}$-time randomized exact algorithm \cite{DBLP:conf/aistats/Chierichetti0LV19}. The running time of this algorithm  has a ``bad'' exponential dependence on the number of groups, i.e., the running time is not a \emph{fixed} polynomial in $n$. Thus, this algorithm quickly becomes impractical if  $\ell$ grows. Our research on proportionally fair matching is driven by the following question. 
\emph{Do there exist efficient algorithms with guaranteed performance for proportionally fair matching when the number of groups $\ell$ is more than two?}

\subsection{Our results and contributions } In this work, we obtain several results on the \probpropfair problem with any arbitrary $\ell$ number of groups. 
\begin{itemize}
    \item First, we show that the problem is extremely hard for any general $\ell$ number of groups, in the sense that it is not possible to obtain any approximation algorithm in $2^{o(\ell)}n^{O(1)}$ time even on path graphs, unless the Exponential Time Hypothesis (ETH) \cite{impagliazzo2001complexity} is false. 
\item To complement our hardness result, we design a  $1/4\ell$-approximation algorithm that runs in $2^{O(\ell)}n^{O(1)}$ time. Our algorithm might violate the lower ($\alpha$) and upper ($\beta$) bounds by at most a multiplicative factor of $(1+4\ell/|\text{OPT}|)$ if $|$OPT$|$ is more than $4{\ell}^2$, where OPT is any  optimum solution. Thus, the violation factor is at most $1+1/\ell$, and tends to 1 with asymptotic values of $|$OPT$|$. 
\item We also consider a restricted case of the problem, referred to as the $\beta$-limited case in \cite{DBLP:conf/aistats/Chierichetti0LV19}, where we only have the upper bound, i.e., no edges might be present from some groups. In this case, we could improve the approximation factor to $1/2\ell$ and running time to polynomial. 
\item Lastly, we show that the parameterized version of the problem where one seeks for a proportionally fair matching of size $k$, can be solved exactly
in $2^{O(k)} n^{O(1)}$ time. Thus the problem is fixed-parameter tractable parameterized by $k$.
\end{itemize}
All of our algorithms are based on simple schemes. Our approximation algorithms use an iterative peeling scheme that in each iteration, extracts a \textit{rainbow} matching containing at most one edge from every group. The exact algorithm is based on a non-trivial application of the celebrated color-coding scheme \cite{alon1995color}. These algorithms appear in Sections \ref{sec:prop-approx},  \ref{sec:beta-limited}, and  \ref{sec:exact}, respectively. The hardness proof is given in Section \ref{sec:hardness}.


\subsection{Related work} In recent years, researchers have introduced and studied several different notions of fairness, e.g., disparate impact \cite{feldman2015certifying}, statistical parity \cite{DBLP:conf/kdd/ThanhRT11,DBLP:conf/icdm/KamishimaAS11}, individual fairness \cite{dwork2012fairness} and group fairness \cite{dwork2018group}. Kleinberg et al. \cite{kleinberg2017inherent} formulated three notions of fairness and showed that it is theoretically impossible to satisfy them simultaneously. See also \cite{DBLP:conf/kdd/Corbett-DaviesP17,chouldechova2017fair} for similar exposures. 

The notion of proportional fairness with multiple protected groups is widely studied in the literature, which is based on disparate impact \cite{feldman2015certifying}. Bei~et al.~\cite{BeiLPW20} studied the \textit{proportional candidate selection problem}, where the goal is to select a subset of candidates with various attributes from a given set while satisfying certain proportional fairness constraints. Goel~et al.~\cite{GoelYF18} considered the problem of learning non-discriminatory and proportionally fair classifiers and proposed the \textit{weighted sum of logs} technique. Proportional fairness has also been considered in the context of Federated learning \cite{zhang2022equality}. Additionally, proportional fairness has been studied in the context of numerous optimization problems including voting \cite{ebadian2022optimized}, scheduling \cite{kesavan2022proportional,lu2022optimization}, Kidney exchange \cite{st2022adaptation}, and Traveling Salesman Problem \cite{nguyen2022nash}. 

Several different fair matching problems have been studied in the literature. \cite{DBLP:journals/algorithmica/HuangKM016} studied fair $b$-matching, where matching preferences for each vertex are given as ranks, and the goal is to avoid assigning vertices to high ranked preferences as much as possible. Fair-by-design-matching is studied in \cite{DBLP:journals/datamine/Garcia-SorianoB20}, where instead of a single matching, a probability distribution over all feasible matchings is computed which guarantees individual fairness. See also \cite{kamada2020fair,klaus2006procedurally}. 


Apart from the fair versions of matchings, many constrained versions are also studied \cite{DBLP:conf/mfcs/Stamoulis14,DBLP:journals/mp/BergerBGS11}. \cite{DBLP:conf/mfcs/Stamoulis14} studied the Bounded Color Matching (BCM) problem where edges are colored and from each color class, only a given number of edges can be chosen. BCM is a special case of 3-set packing and, hence, admits a $3/4$-approximation \cite{DBLP:conf/mfcs/Stamoulis14}. We note that the $\beta$-limited case of \probpropfair is a special case of BCM and, thus, a $3/4$-approximation follows in this case where the upper bound might be violated by $3/4$ factor. One should compare this factor with our violation factor, which asymptotically tends to 1.

%% file: prelims.tex
\section{Preliminaries}
\label{sec:prelims}
For an integer $\ell \ge 1$, let $[\ell] \coloneqq \{1, 2, \ldots, \ell\}$.
Consider any undirected $n$-vertex graph $G=(V,E)$ such that the edges in $E$ are colored by colors in $C=\{1,\ldots,\ell\}$. The function $\chi: E \rightarrow C$ describes the color assignment. For each color $c \in C$, let $E_c$ be the set of edges colored by the color $c$, i.e., $E_c={\chi}^{-1}(c)$. A subset $E'\subseteq E$ is a \textit{matching} in $G$ if no two edges in $E'$ share a common vertex. 
\begin{definition}
\textbf{$(\alpha,\beta)$-balanced matching}. Given $0\le \alpha \le \beta \le 1$, a matching $M\subseteq E$ is called $(\alpha,\beta)$-balanced if for each color $c\in C$, we have that \quad $\displaystyle \alpha \le \frac{|M\cap E_c|}{|M|}\le \beta.$    
\end{definition}
Thus a matching is $(\alpha,\beta)$-balanced if it contains at least $\alpha$ and at most $\beta$ fraction of edges from every color. In the \probpropfair problem, the goal is to find a maximum-sized $(\alpha,\beta)$-balanced matching. In the restricted $\beta$-limited case of the problem, $\alpha=0$, i.e., we only have the upper bound.

For $\gamma \le 1$ and $\Delta\ge 1$, a $(\gamma,\Delta)$-approximation algorithm for \probpropfair computes a matching of size at least $\gamma\cdot |$OPT$|$, where every color appears in at least ${\alpha}/\Delta$ fraction of the edges and in at most $\beta\cdot \Delta$ fraction. OPT is an optimum $(\alpha,\beta)$-balanced matching. 

A matching is called a \textit{rainbow matching} if all of its edges have distinct colors. We will need the following result due to Gupta et al.~\cite{GuptaRSZ19}. 
\begin{theorem}\label{thm:rainbow}
For some integer $k > 0$, suppose there is a rainbow matching in $G$ of size $k$. There is a $2^k\cdot n^{O(1)}$ time algorithm that computes a rainbow matching of size $k$. 
\end{theorem}

%% file: proportionally.tex
\section{A $(\frac{1}{4\ell},1+\frac{4\ell}{|OPT|})$-Approximation for \probpropfair}
\label{sec:prop-approx}

In this section, we design an approximation algorithm for \probpropfair. Let OPT be an optimum $(\alpha,\beta)$-balanced matching, ${\text{OPT}}_c=\text{OPT}\cap E_c$. We design two algorithms: one for  the case when $\alpha > 0$ and the other for the complementary $\beta$-limited case. In this section, we slightly abuse the notation, and use OPT (resp.\ OPT$_c$ for some color $c \in C$) to refer to $|$OPT$|$ (resp.\ $|$OPT$_c|$). The intended meaning should be clear from the context; however we will disambiguate in case there is a possibility of confusion.

First, we consider the $\alpha > 0$ case. Immediately, we have the following observation. 

\begin{observation} \label{obs:rainbowexists}
For any color $c\in C$, \emph{OPT} contains at least one edge of color $c$ and, hence, $G$ contains a rainbow matching of size $\ell$.  
\end{observation}



Our algorithm runs in \textit{rounds}. In the following, we define a round. The input in each round is a subgraph $G'=(V',E')$ of $G$.  

\medskip\noindent\textbf{Round.} Initially $M=\emptyset$. For every color $1\le c\le \ell$, do the following in an iterative manner. If there is no edge of color $c$ in $G'$, go to the next color or terminate and return $(G',M)$ if $c=\ell$. Otherwise, pick any edge $e$ of color $c$ from $G'$ and add $e$ to the already computed matching  $M$. Remove all the edges  (including $e$) from $G'$ that share a common vertex with $e$. Repeat the process for the next color with the current (or updated) graph $G'$ or terminate and return $(G',M)$ if $c=\ell$. 

Thus in each round, we try to pick a rainbow matching in a greedy manner. Next, we describe our algorithm. The most challenging part of our algorithm is to ensure that the final matching computed is $(\alpha,\beta)$-balanced modulo a small factor, i.e., we need to ensure both the lower and the upper bounds within a small factor for each color. Note that just the above greedy way of picking edges might not even ensure that at least one edge is selected from each color. We use the algorithm of \cite{GuptaRSZ19} in the beginning to overcome this barrier. However, the rest of our algorithm is extremely simple.

\medskip\noindent\textbf{The Algorithm.} We assume that we know the size of OPT. We describe later how to remove this assumption. Apply the algorithm in Theorem \ref{thm:rainbow} on $G$ to compute a rainbow matching $M'$ of size $\ell$. If OPT $\le 4\ell^2$, return $M\coloneqq M'$ as the solution and terminate. Otherwise, remove all the edges of $M'$ and the edges adjacent to them from $G$ to obtain the graph $G_0$. Initialize $M$ to $M'$. Greedily pick matched edges in rounds using the \textbf{Round} procedure and add them to $M$ until exactly $\lceil\text{OPT}/(4\ell)\rceil$ edges are picked in total. In particular, the graph $G_{0}$ is the input to the $1$-st round and $G_1$ is the output graph of the 1-st round. $G_1$ is the input to the 2-nd round and $G_2$ is the output graph of the 2-nd round, and so on. Note that it might be the case that the last round is not completed fully if the size of $M$ is reached to $\lceil\text{OPT}/(4\ell)\rceil$ before the completion of the round. 

Note that the above algorithm is oblivious to $\alpha$ and $\beta$ in the sense that it never uses these values. Nevertheless, we prove that the computed matching is $(\alpha,\beta)$-balanced modulo a small factor. Now we analyze our algorithm. 

\subsection{The Analysis}

Let $M_c=M\cap E_c$. Also, let $c^*$ be a color $c\in C$ such that $|\text{OPT}_{c}|$ is the minimum at $c=c^*$. 

\begin{observation}\label{obs:ab}
 $\alpha \le 1/\ell\le \beta$. 
\end{observation}
\begin{proof}
	Let $\hat{c}$ be a color $c\in C$ such that $|\text{OPT}_{c}|$ is the minimum at $c=\hat{c}$. 
	By definition, OPT $\ge \ell\cdot \text{OPT}_{c^*}$, or $\text{OPT}_{c^*}/\text{OPT}\le 1/\ell$. Thus, $\alpha \le \text{OPT}_{c^*}/\text{OPT} \le 1/\ell$. Similarly, OPT $\le \ell\cdot \text{OPT}_{\hat{c}}$, or $\text{OPT}_{\hat{c}}/\text{OPT}\ge 1/\ell$. Thus, $\beta \ge \text{OPT}_{\hat{c}}/\text{OPT} \ge 1/\ell$.
\end{proof}

First we consider the case when OPT $ \le 4\ell^2$. In this case the returned matching $M$ is a rainbow matching of size exactly $\ell$. The existence of such a matching follows by Observation \ref{obs:rainbowexists}. 
Thus, we immediately obtain a $4\ell$-approximation. As $|M_c|/|M|=1/\ell$ in this case, by Observation \ref{obs:ab}, $\alpha \le |M_c|/|M|\le \beta$. Thus we obtain the desired result. In the rest of the proof, we analyze the case when OPT $> 4\ell^2$. We start with the following lemma. 



\begin{lemma}\label{lem:matchingsize}
The algorithm successfully computes a matching of size exactly $\lceil\emph{OPT}/(4\ell)\rceil$. Moreover, for each color $c$ with $\emph{OPT}_c > 4\ell$ and round $i \in [1,\lceil\emph{OPT}_c/(4\ell)\rceil-1]$, $G_{i-1}$ contains an edge of color $c$.  
\end{lemma}

\begin{proof}
Note that by Observation \ref{obs:rainbowexists}, the algorithm in Theorem \ref{thm:rainbow} successfully computes a rainbow matching $M'$ of size $\ell$. Now consider any color $c$ such that $\text{OPT}_c \le 4\ell$. For such a color, $M$ already contains at least $1 \ge \lceil \text{OPT}_c/(4\ell)\rceil$ edge. Now consider any other color $c$ with $|\text{OPT}_c| > 4\ell$. Consider the rainbow matching $M'$ computed in the beginning. As $|M'|= \ell$, the edges of $M'$ can be adjacent to at most $2\ell$ edges from OPT, since it is a matching. In particular, the edges of $M'$ can be adjacent to at most $2\ell$ edges from the set $\text{OPT}_c$. Hence, $G_0$ contains at least $\text{OPT}_c-2\ell$ edges of the set $\text{OPT}_c$. Now consider the execution of round $i\ge 1$. At most $\ell$ edges are chosen in this round. Hence, these edges can be adjacent to at most $2\ell$ edges of $\text{OPT}_c$. It follows that at most $2\ell$ fewer edges of the set $\text{OPT}_c$ are contained in $G_i$ compared to $G_{i-1}$. As $G_0$ has at least $\text{OPT}_c-2\ell$ edges from the set $\text{OPT}_c$ of color $c$ and $\text{OPT}_c > 4\ell$, for each of the first $\lceil(\text{OPT}_c-2\ell)/(2\ell)\rceil=\lceil\text{OPT}_c/(2\ell)\rceil-1$ rounds, the algorithm will be able to pick an edge of color $c$. Thus from such a color $c$ with $\text{OPT}_c > 4\ell$, it can safely pick at least $\lceil\text{OPT}_c/(2\ell)\rceil\ge \lceil\text{OPT}_c/(4\ell)\rceil$ edges in total. Now, as $\text{OPT}=\sum_c \text{OPT}_c$,
$\sum_{c\in C}  \lceil\text{OPT}_c/(4\ell)\rceil\ge \lceil\text{OPT}/(4\ell)\rceil$.
 It follows that the algorithm can pick at least $\lceil\text{OPT}/(4\ell)\rceil$ edges. As we stop the algorithm as soon as the size of $M$ reaches to $\lceil\text{OPT}/(4\ell)\rceil$, the lemma follows.       
\end{proof}

Note that the claimed approximation factor trivially follows from the above lemma. Next, we show that $M$ is $(\alpha,\beta)$-balanced modulo a small factor that asymptotically tends to 1 with the size of $\text{OPT}$. 

\begin{lemma}\label{lem:lowerbnd}
For each color $c\in C$, $|M_c|\ge |\emph{\text{OPT}}_{c^*}|/(4\ell)$. 
\end{lemma}

\begin{proof}
If $\text{OPT}_{c^*}\le 4\ell$, $|M_c| \ge 1 \ge \text{OPT}_{c^*}/(4\ell)$. So, assume that  $\text{OPT}_{c^*}> 4\ell$. Now suppose $|M_c|<  \text{OPT}_{c^*}/(4\ell)$ for some $c$. By Lemma \ref{lem:matchingsize}, for each of the first $\lceil\text{OPT}_c/(4\ell)\rceil-1\ge \lceil\text{OPT}_{c^*}/(4\ell)\rceil-1$ rounds, $G_{i-1}$ contains an edge of color $c$. It follows that the algorithm was forcibly terminated in some round $i \le (\text{OPT}_{c^*}/(4\ell))-1$. Thus, the number of edges chosen from each color $c'\ne c$ is at most $\text{OPT}_{c^*}/(4\ell)$.
Hence,
\begin{align*}
|M|  &=\sum_{c'\ne c}|M_{c'}|+|M_c| \\&< (\ell-1)\cdot (\text{OPT}_{c^*}/(4\ell))+(\text{OPT}_{c^*}/(4\ell))\\&  \le \lceil\text{OPT}/(4\ell)\rceil.
\end{align*}
This contradicts Lemma \ref{lem:matchingsize}, which states that we select exactly $\lceil \text{OPT}/(4\ell)\rceil$ edges.  
\end{proof}

\begin{corollary} \label{cor:lowerbound}
For each color $c\in C$, $(|M_c|/|M|) \ge \frac{\alpha}{(1+4\ell/\emph{\text{OPT}})}$.  
\end{corollary}
\begin{proof}
	By Lemma \ref{lem:lowerbnd}, $|M_c|\ge \text{OPT}_{c^*}/(4\ell)$.
	\begin{align*}
		\frac{|M_c|}{|M|}\ge\frac{(\text{OPT}_{c^*}/(4\ell))}{\lceil\text{OPT}/(4\ell)\rceil}\ge\frac{(\text{OPT}_{c^*}/(4\ell))}{(\text{OPT}/(4\ell))+1}\\
		=\frac{(\text{OPT}_{c^*})/(\text{OPT})}{(1+4\ell/\text{OPT})}\ge \frac{\alpha}{(1+4\ell/\text{OPT})}.
	\end{align*}
	The last inequality follows as \text{OPT} satisfies the lower bound for all colors. 
\end{proof}

Now we turn to proving the upper bound. Let ${\alpha}^*=\text{OPT}_{c^*}/\text{OPT}$. 
\begin{lemma}\label{lem:upperbound1}
For each color $c\in C$, $|M_c|\le \frac{\beta}{{\alpha}^*}\cdot (\emph{\text{OPT}}_{c^*}/(4\ell))+1$.
\end{lemma}

\begin{proof}
Suppose for some $c\in C$, $|M_c|> \frac{\beta}{{\alpha}^*}\cdot (\text{OPT}_{c^*}/(4\ell))+1$. Then the number of rounds is strictly greater than $\frac{\beta}{{\alpha}^*}\cdot (\text{OPT}_{c^*}/(4\ell))$. Now, for any $c'$, $\text{OPT}_{c'} \ge  {\alpha}^* \cdot \text{OPT}$ and $\text{OPT}_{c'} \le  \beta \cdot \text{OPT}$. Thus, by the definition of ${\alpha}^*$, $\frac{\beta}{{\alpha}^*}\cdot \text{OPT}_{c^*}\ge \text{OPT}_{c'}$. It follows that, for each $c'$, the number of rounds is strictly greater than $\text{OPT}_{c'}/(4\ell)$. Hence, for each $c'\in C$, more than $(\text{OPT}_{c'}/(4\ell))+1$ edges have been chosen. Thus, the total number of edges chosen is strictly larger than \[\sum_{c'\in C} ((\text{OPT}_{c'}/(4\ell))+1) \ge  \lceil \text{OPT}/(4\ell)\rceil.\] This contradicts Lemma \ref{lem:matchingsize}, which states that we select exactly $\lceil \text{OPT}/(4\ell)\rceil$ edges.  
\end{proof}

\begin{corollary} \label{cor:upperbound}
For each color $c\in C$, $(|M_c|/|M|) \le \beta\cdot (1+\frac{4\ell}{\emph{OPT}})$.  
\end{corollary}
\begin{proof}
	By Lemma \ref{lem:upperbound1}, 
	\begin{align*}
		\frac{|M_c|}{|M|} & \le \frac{({\beta}/{{\alpha}^*})\cdot (\text{OPT}_{c^*}/(4\ell))+1}{\lceil\text{OPT}/(4\ell)\rceil}\\
		&\le \frac{({\beta}/{{\alpha}^*})\cdot (\text{OPT}_{c^*}/(4\ell))+({\beta}/{{\alpha}^*})}{\text{OPT}/(4\ell)}\\
		& =\frac{\beta}{{\alpha}^*}\cdot \frac{\text{OPT}_{c^*}}{\text{OPT}}\cdot  \bigg(1+\frac{4\ell}{\text{OPT}}\bigg)\\
		&= \frac{\beta}{{\alpha}^*}\cdot{\alpha}^*\bigg(1+\frac{4\ell}{\text{OPT}}\bigg)= \beta\cdot \bigg(1+\frac{4\ell}{\text{OPT}}\bigg).
	\end{align*}
	The second inequality follows, as ${\alpha}^*\le \beta$ or $\beta/{\alpha}^* \ge 1$. 
\end{proof}

Now let us remove the assumption that we know the size of an optimal solution. Note that $\ell\le \text{OPT}\le n$. We probe all values between $\ell$ and $n$, and for each such value $T$ run our algorithm. For each matching $M$ returned by the algorithm, we check whether $M$ is $(\frac{\alpha}{(1+4\ell/T)}, \beta\cdot (1+\frac{4\ell}{T}))$-balanced. If this is the case, then we keep this solution. Otherwise, we discard the solution. Finally, we select a solution of the largest size among the ones not discarded. By the above analysis, with $T=\text{OPT}$, the matching returned satisfies the desired lower and upper bounds, and has size exactly $\lceil\text{OPT}/(4\ell)\rceil$. Finally, the running time of our algorithm is dominated by $2^{\ell} n^{O(1)}$ time to compute a rainbow matching algorithm, as stated in Theorem \ref{thm:rainbow}. 
\begin{theorem}
There is a $2^{\ell}\cdot n^{O(1)}$ time $(1/4\ell,1+{4\ell}/{\emph{OPT}})$-approximation algorithm for \probpropfair with $\alpha > 0$. 
\end{theorem}


\section{A Polynomial-time Approximation in the $\beta$-limited Case} \label{sec:beta-limited}
In the $\beta$-limited case, again we make use of the \textbf{Round} procedure. But, the algorithm is slightly different. Most importantly, we do not apply the algorithm in Theorem \ref{thm:rainbow} in the beginning. Thus, our algorithm runs in polynomial time. 

\medskip\noindent\textbf{The Algorithm.} Assume that we know the size of OPT. If OPT $\le 2\ell$, we pick any edge and return it as the solution. Otherwise,  we just greedily pick matched edges in rounds using the \textbf{Round} procedure with the following two cautions. If for a color, at least $\beta\cdot \text{OPT}/(2\ell)$ edges have already been chosen, do not choose any more edge of that color. If at least $(\text{OPT}/{2\ell})-1$ edges have already been chosen, terminate. 

Now we analyze the algorithm. First note that if  OPT $\le 2\ell$, the returned matching has only one edge. The upper bound is trivially satisfied and also we obtain a $2\ell$-approximation. Henceforth, we assume that OPT $ > 2\ell$.  Before showing the correctness and analysis of the approximation factor, we show the upper bound for each color. Again let $M$ be the computed matching and $M_c=M\cap E_c$. Later we prove the following lemma. 

\begin{lemma}\label{lem:uba0}
	The algorithm always returns a matching of size at least $(\emph{OPT}/{2\ell})-1$. 
\end{lemma}

Assuming this we have the following observation.  

\begin{observation}
	For each color $c\in C$, $|M_c|/|M|\le \beta\cdot (1+\frac{2\ell}{|\emph{OPT}|})$. 
\end{observation}

\begin{proof}
	By Lemma \ref{lem:uba0} and the threshold put on each color in the algorithm, 
	\begin{align*}
		\frac{|M_c|}{|M|}\le \frac{\beta\cdot \text{OPT}/(2\ell)}{(\text{OPT}/{2\ell})-1}\le \beta\cdot (1+\frac{2\ell}{\text{OPT}})
	\end{align*}
	The last inequality follows, as $\text{OPT}> 2\ell$. 
\end{proof}

Next, we prove Lemma \ref{lem:uba0}. 

\begin{proof}
	Let $C_1$ be the subset of colors such that for each $c\in C_1$, the algorithm picks at least $\beta\cdot \text{OPT}/(2\ell)$ edges. Note that the algorithm can terminate in two ways (i) it has already picked at least $(\text{OPT}/{2\ell})-1$ edges, and (ii) all the edges have been exhausted. Note that if (i) happens, then we are done. We prove that (ii) cannot happen without (i). Suppose (ii) happens, but not (i). Let OPT$'$ be the subset of OPT containing edges of colors in $C'=C\setminus C_1$. Recall that $G_{i-1}$ is the input graph to the $i$-th round and $G_i$ is the output graph for $i \ge 1$. The number of edges chosen in $i$-th round is at most $\ell$. Hence, these edges can be adjacent to at most $2\ell$ edges in $G_{i-1}$. In particular, at most $2\ell$ less edges of OPT$'$ are contained in $G_i$ compared to $G_{i-1}$. It follows that the algorithm can pick at least $\lfloor$OPT$'/{2\ell}\rfloor$ edges of colors in $C'$. As for each color in $C'$, less than $\beta\cdot \text{OPT}/(2\ell)$ edges are chosen, the algorithm indeed chooses at least $\lfloor$OPT$'/{2\ell}\rfloor$ edges of these colors. The total number of edges chosen by the algorithm is, 
	\begin{align*}
		\sum_{c\in C_1} |M_c|+\sum_{c\in C'} |M_c| & \ge \sum_{c\in C_1} \beta\cdot \text{OPT}/(2\ell) + \lfloor\text{OPT}'/{2\ell}\rfloor  \\ 
		& \ge \sum_{c\in C_1} \text{OPT}_c/{2\ell} + \lfloor\text{OPT}'/{2\ell}\rfloor\\ & \ge  (\text{OPT}/{2\ell})-1
	\end{align*}
	But, this is a contradiction to our assumption, and hence the lemma follows. 
\end{proof}

\begin{theorem}
	There is a polynomial time algorithm for \probpropfair in the $\beta$-limited case that returns a matching of size at least $(\emph{OPT}/{2\ell})-1$ where every color appears in at most $\beta\cdot (1+{2\ell}/{\emph{OPT}})$ fraction of the edges. 
\end{theorem}


\input{fpt}


\section{Hardness of Approximation for \textsc{Proportionally Fair Matching}} \label{sec:hardness}
In this section, we show  an inapproximability result for \probpropfair under the Exponential Time Hypothesis (ETH) \cite{impagliazzo2001complexity}. ETH states that  $2^{\Omega(n)}$ time is needed to solve any generic 3SAT instance with $n$ variables. For our purpose, we need the following restricted version of 3SAT. 

\medskip\noindent\textbf{3SAT-3}\\
INPUT: Set of clauses $T=\{C_1,\ldots,C_m\}$ in variables $x_1,\ldots,x_n$, each clause being the disjunction of 3 or 2 literals, where a literal is a variable $x_i$ or its negation $\bar{x}_i$. Additionally, each variable appears 3 times.\\
QUESTION: Is there a truth assignment that simultaneously satisfies all the clauses? 

3SAT-3 is known to be \np-hard \cite{Yannakakis78}. We need the following stronger lower bound for 3SAT-3 proved in \cite{cygan2017hitting}.

\begin{proposition}[\cite{cygan2017hitting}]\label{prop:3sat3}
	Under ETH, 3SAT-3 cannot be solved in $2^{o(n)}$ time. 
\end{proposition}


We reduce 3SAT-3 to \probpropfair which rules out any approximation for the latter problem in $2^{o(\ell)}n^{O(1)}$ time. Our reduction is as follows. For each clause $C_i$, we have a color $i$. Also, we have $n-1$ additional colors ${m+1},\ldots,{m+n-1}$. Thus, the set of colors $C=\{1,\ldots,{m+n-1}\}$. For each variable $x_i$, we construct a gadget, which is a 3-path (a path with 3 edges). Note that $x_i$ can either appear twice in its normal form or in its negated form, as it appears 3 times in total. Let $C_{i_1}, C_{i_2}$ and $C_{i_3}$ be the clauses where $x_i$ appears. Also, suppose it appears in $C_{i_1}$ and $C_{i_3}$ in one form, and in $C_{i_2}$ in the other form. We construct a 3-path $P_i$ for $x_i$ where the $j$-th edge has color $i_j$ for $1\le j\le 3$. Additionally, we construct $n-1$ 3-paths $Q_{i,i+1}$ for $1\le i\le n-1$. All edges of $Q_{i,i+1}$ is of color $m+i$. Finally, we glue together all the paths in the following way to obtain a single path. For each $1\le i\le n-1$, we glue $Q_{i,i+1}$ in between $P_i$ and $P_{i+1}$ by identifying the last vertex of $P_i$ with the first vertex of $Q_{i,i+1}$ and the last vertex of $Q_{i,i+1}$ with the first vertex of $P_{i+1}$. Thus we obtain a path $P$ with exactly $3(n+n-1)=6n-3$ edges. Finally, we set $\alpha = \beta = 1/(m+n-1)$. 

\begin{lemma}
\label{lem:3sat3iffab}
	There is a satisfying assignment for the clauses in 3SAT-3 if and only if there is an $(\alpha,\beta)$-balanced matching of size at least $m+n-1$.  
\end{lemma}
\begin{proof}
	Suppose there is a satisfying assignment for all the clauses. For each clause $C_j$, consider a variable, say $x_i$, that satisfies $C_j$. Then there is an edge of color $j$ on $P_i$. Add this edge to a set $M$. Thus, after this step, $M$ contains exactly one edge of color $j$ for $1\le j\le m$. Also, note that for each path $P_i$, if the middle edge is chosen, then no other edge from $P_i$ can be chosen. This is true, as the variable $x_i$ can either satisfy the clauses where it appears in its normal form or the clauses where it appears in its negated form, but not both types of clauses. Hence, $M$ is a matching. Finally, for each path $Q_{i,i+1}$, we add its middle edge to $M$. Note that $M$ still remains a matching. Moreover, $M$ contains exactly one edge of color $j$ for $1\le j\le m+n-1$. As $\alpha = \beta = 1/(m+n-1)$, $M$ is an $(\alpha,\beta)$-balanced matching. 
	
	Now suppose there is an $(\alpha,\beta)$-balanced matching $M$ of size at least $m+n-1$. First, we show that $|M|=m+n-1$. Note that if $|M| > m+n-1$, then the only possibility is that $|M|=2(m+n-1)$, as $\alpha=\beta$ and at most 2 edges of color $j$ can be picked in any matching for $m+1\le j\le m+n-1$. Suppose $|M|=2(m+n-1)$. Then from each $Q_{i,i+1}$, $M$ contains the first and the third edge. This implies, from each $P_t$, $1\le t\le n$, we can pick at most one edge. Thus, total number of edges in $M$ is at most $2(n-1)+n$. It follows that $2m+2n-2\le 2n-2+n$ or $n\ge 2m $. Now, in 3SAT-3 the total number of literals is $3n$ and at most $3m$, as each variable appears 3 times and each clause contains at most 3 literals. This implies $n \le m$, and we obtain a contradiction. Thus, $|M|=m+n-1$. 
	Now, consider any $P_i$. 
	In the first case, the first and third edges of $P_i$ are corresponding to literal $x_i$ and, hence, the middle edge is corresponding to the literal $\bar{x}_i$. If the middle edge is in $M$, assign 0 to $x_i$, otherwise, assign 1 to $x_i$. In the other case, if the middle edge is in $M$, assign 1 to $x_i$, otherwise, assign 0 to $x_i$. We claim that the constructed assignment satisfies all the clauses. Consider any clause $C_j$. Let $e \in P_i$ be the edge in $M$ of color $j$ for $1\le j\le m$. Note that $e$ can be the middle edge in $P_i$ or not. In any case, if $e$ is corresponding to $\bar{x}_i$, we assigned 0 to $x_i$, and if $e$ is corresponding to ${x}_i$, we assigned 1 to $x_i$. Thus, in either case, $C_j$ is satisfied. This completes the proof of the lemma.      
\end{proof}

Note that for a 3SAT-3 instance the total numbers of literals is $3n$. As each clause contains at least 2 literals, $m \le 3n/2$. Now, for the instances constructed in the above proof,  the number of colors $\ell= m+n-1\le 3n/2+n-1=5n/2-1$. Thus, the above lemma along with Proposition \ref{prop:3sat3} show that it is not possible to decide whether there is an $(\alpha,\beta)$-balanced matching of a given size in time $2^{o(\ell)}n^{O(1)}$. Using this, we also show that even no $2^{o(\ell)}n^{O(1)}$ time approximation algorithm is possible. Suppose there is a $2^{o(\ell)}n^{O(1)}$ time $\gamma$-approximation algorithm, where $\gamma < 1$. For our constructed path instances, we apply this algorithm to obtain a matching. Note that the $\gamma$-approximate solution $M$ must contain at least one edge of every color, as $\alpha=\beta$. By the proof in the above lemma, $|M|$ is exactly $m+n-1$. Hence, using this algorithm, we can decide in $2^{o(\ell)}n^{O(1)}$ time whether there is an $(\alpha,\beta)$-balanced matching of size $m+n-1$. But, this is a contradiction, which leads to the following theorem.

\begin{theorem}
	For any $\gamma > 1$, under ETH, there is no $2^{o(\ell)}n^{O(1)}$ time $\gamma$-approximation algorithm for \probpropfair, even on paths. 
\end{theorem}

%% file: fpt.tex
\section{An Exact Algorithm for \probpropfair}\label{sec:exact}

\begin{theorem}
	There is a $2^{O(k)} n^{O(1)}$-time algorithm that either finds a solution of size $k$ for a \probpropfair instance, or determines that none exists.
	\label{thm:fpt}
\end{theorem}
\begin{proof}
	We present two different algorithms using the well-known technique of color coding: one for the case $\alpha = 0$ ($\beta$-limited case), and one for the case $\alpha > 0$.
	\vspace{-0.2cm}
	\subparagraph*{$\beta$-limited case.} 
	We aim to reduce the problem to finding a rainbow matching of size $k$, which we then solve via Theorem~\ref{thm:rainbow}.
	The graph $G$ remains the same, however the coloring is going to be different. Namely, for each of the original colors $c \in C$
	we color the edges in $E_c$ uniformly and independently at random from a set of $k'$ new colors, where $k' = \lfloor\beta k\rfloor$.
	Thus, the new instance $I'$ is colored in  $\ell \cdot k'$ colors. We use the algorithm of Theorem~\ref{thm:rainbow} to find a rainbow matching of size $k$ in the colored graph in $I'$. Clearly, if a rainbow matching $M$ of size $k$ is found, then the same matching $M$ is a $\beta$-limited matching of size $k$ in the original coloring. This holds since by construction for any original color $c \in C$, there are $k'$ new colors in the edge set $E_c$, and therefore no more than $k'$ edges in $|M \cap E_c|$.
	
	In the other direction, we show that if there exists a $\beta$-limited matching $M$ of size $k$ with respect to the original coloring, then with good probability $M$ is a rainbow matching of size $k$ in the new coloring. Assume the original colors $c_1$, \ldots, $c_t$, for some $1 \le t \le \ell$, have non-empty intersection with $M$, and for each $j \in [t]$ denote $k_j = |M \cap E_{c_j}|$. Observe that $\sum_{j = 1}^t k_j = k$, and for each $j \in [t]$, $1 \le k_j \le k'$.
	\begin{claim}\label{cl:probability}
		There exists some $\delta > 0$ such that for each $j \in [t]$:
		\begin{equation}
			\prob \left[M \cap \left(\bigcup_{i = 1}^j E_{c_i} \right) \text{ is a rainbow matching in } I'\right] \nonumber \ge \exp\lr{-\delta \sum_{i = 1}^j k_i},
			\label{eq:color_prob}
		\end{equation}
	\end{claim}
	\begin{proof}
		We prove the claim by induction on $j$. For the base case, clearly (\ref{eq:color_prob}) holds for $j = 0$. Now, fix $j \in [t]$ and assume the statement holds for each $j' < j$, we show that (\ref{eq:color_prob}) also holds for $j$. Consider the $k_j$ edges of $M \cap E_{c_j}$, they are colored uniformly and independently in $k' \ge k_j$ colors. By counting possible colorings  of $M \cap E_{c_j}$, it follows that
		\begin{align*}
			\prob \left[M \cap E_{c_j} \text{ is a rainbow matching }\right] \\
			\ge \frac{(k')! / (k' - k_j)!}{(k')^{k_j}}\ge \frac{k_j!}{k_j^{k_j}} \ge 2^{-\delta k_j},
		\end{align*}
		where the last bound is by Stirling's formula. Now, since colors used for $E_{c_{j}}$ do not appear anywhere else, using the inductive hypothesis we get
		\begin{align*}
			&\prob \left[M \cap \left(\bigcup_{i = 1}^j E_{c_i} \right) \text{ is a rainbow matching }\right] \\= & \prob \left[M \cap \left(\bigcup_{i = 1}^{j - 1} E_{c_i} \right) \text{ is a rainbow matching }\right] 
			\\&\cdot\prob \left[M \cap E_{c_j} \text{ is a rainbow matching }\right]
			\\\ge & 2^{-\delta \sum_{i = 1}^{j - 1} k_i} \cdot 2^{-\delta k_j} = 2^{-\delta \sum_{i = 1}^{j} k_i}.
		\end{align*}
	\end{proof}

	Applying (\ref{eq:color_prob}) with $j = t$, we obtain that $M$ is a rainbow matching with probability at least $2^{-\delta k}$. By repeating the reduction above $2^{O(k)}$ times independently, we achieve that the algorithm succeeds with constant probability.
	\vspace{-0.2cm}
	\subparagraph*{The case $\alpha > 0$.}
	We observe that in this case, if a matching is fair it necessarily contains at least one edge from each of the groups. Thus, if the number of groups $\ell$ is greater than $k$, we immediately conclude there cannot be a fair matching of size $k$. Otherwise, we guess how the desired $k$ edges are partitioned between the $\ell$ groups $C = \{c_1, \ldots, c_\ell\}$. That is, we guess the numbers $k_j$ for $j \in [\ell]$ such that 
	$\sum_{j = 1}^\ell k_j = k$, and $\alpha k \le k_j \le \beta k$ for each $j \in [\ell]$. From now on, the algorithm is very similar to the $\beta$-limited case. For each group $c_j$, we color the edges of $E_{c_j}$ from a set of $k_j$ colors  uniformly and independently at random, where the colors used for each $E_{c_j}$ are non-overlapping.
	Now we use the algorithm of Theorem~\ref{thm:rainbow} to find a rainbow matching of size $k$. If there is a rainbow matching $M$ of size $k$, the same matching is a fair matching of size $k$ for the original instance, since in each $E_{c_j}$ exactly $k_j$ edges are chosen, which is at least $\alpha k$ and at most $\beta k$. In the other direction, if there is a fair matching $M$ of size $k$ in the original instance, by \eqref{eq:color_prob} the matching $M$ is a rainbow matching in the new instance with probability at least $2^{-\delta k}$. Again, by repeating the coloring subprocess independently $2^{O(k)}$ times, we achieve a constant probability of success. Since there are $2^{O(k)}$ options for partitioning $k$ edges into at most $\ell\le k$ groups, the running time of the whole algorithm is $2^{O(k)} n^{O(1)}$.
	
	Finally, we note that the coloring part in both cases can be derandomized in the standard fashion by using perfect hash families~\cite{NaorSS95}, leading to a completely deterministic algorithm.
\end{proof}


%% file: conclusion.tex
\section{Conclusions} 
\label{sec:conclude} 
In this paper, we study the notion of proportional fairness in the context of matchings in graphs, which has been studied by Chierichetti et al.\ \cite{chierichetti2017fair}. We obtained approximation and exact algorithms for the proportionally fair matching problem. We also complement these results by showing hardness results. It would be interesting to obtain a $o(\ell)$-  or a true $O(\ell)$-approximation for \probpropfair improving our result. As evident from our hardness result, there is a lower bound of $2^{\Omega(\ell)}n^{O(1)}$ on the running time of such an algorithm. 

\paragraph{Acknowledgments.} Most of this work was done when all four authors were affiliated with University of Bergen, Norway. The research leading to these results has received funding from the Research Council of Norway via the project MULTIVAL, and the European Research Council (ERC) via grant LOPPRE, reference 819416.

%% file: mainpaper.bbl
\begin{thebibliography}{10}

\bibitem{agarwal2019fair}
{\sc A.~Agarwal, M.~Dudik, and Z.~S. Wu}, {\em Fair regression: Quantitative
  definitions and reduction-based algorithms}, in International Conference on
  Machine Learning, PMLR, 2019, pp.~120--129.

\bibitem{DBLP:conf/ijcai/AhmadiADFK20}
{\sc S.~Ahmadi, F.~Ahmed, J.~P. Dickerson, M.~Fuge, and S.~Khuller}, {\em An
  algorithm for multi-attribute diverse matching}, in Proceedings of the
  Twenty-Ninth International Joint Conference on Artificial Intelligence,
  {IJCAI} 2020, C.~Bessiere, ed., ijcai.org, 2020, pp.~3--9.

\bibitem{alon1995color}
{\sc N.~Alon, R.~Yuster, and U.~Zwick}, {\em Color-coding}, Journal of the ACM
  (JACM), 42 (1995), pp.~844--856.

\bibitem{angwin2019machine}
{\sc J.~Angwin, J.~Larson, S.~Mattu, and L.~Kirchner}, {\em Machine bias:
  There's software used across the country to predict future criminals. and
  it's biased against blacks}, ProPublica,  (May 23, 2016).

\bibitem{DBLP:journals/corr/abs-2007-10137}
{\sc S.~Bandyapadhyay, F.~V. Fomin, and K.~Simonov}, {\em On coresets for fair
  clustering in metric and euclidean spaces and their applications}, CoRR,
  abs/2007.10137 (2020).

\bibitem{BeiLPW20}
{\sc X.~Bei, S.~Liu, C.~K. Poon, and H.~Wang}, {\em Candidate selections with
  proportional fairness constraints}, in Proceedings of the 19th International
  Conference on Autonomous Agents and Multiagent Systems, {AAMAS} '20,
  Auckland, New Zealand, May 9-13, 2020, A.~E.~F. Seghrouchni, G.~Sukthankar,
  B.~An, and N.~Yorke{-}Smith, eds., International Foundation for Autonomous
  Agents and Multiagent Systems, 2020, pp.~150--158.

\bibitem{bera2019fair}
{\sc S.~Bera, D.~Chakrabarty, N.~Flores, and M.~Negahbani}, {\em Fair
  algorithms for clustering}, in Advances in Neural Information Processing
  Systems, 2019, pp.~4954--4965.

\bibitem{DBLP:journals/mp/BergerBGS11}
{\sc A.~Berger, V.~Bonifaci, F.~Grandoni, and G.~Sch{\"{a}}fer}, {\em Budgeted
  matching and budgeted matroid intersection via the gasoline puzzle}, Math.
  Program., 128 (2011), pp.~355--372.

\bibitem{DBLP:journals/corr/BerkHJJKMNR17}
{\sc R.~Berk, H.~Heidari, S.~Jabbari, M.~Joseph, M.~J. Kearns, J.~Morgenstern,
  S.~Neel, and A.~Roth}, {\em A convex framework for fair regression}, CoRR,
  abs/1706.02409 (2017).

\bibitem{BuolamwiniG18}
{\sc J.~Buolamwini and T.~Gebru}, {\em Gender shades: Intersectional accuracy
  disparities in commercial gender classification}, in Conference on Fairness,
  Accountability and Transparency, {FAT} 2018, 23-24 February 2018, New York,
  NY, {USA}, S.~A. Friedler and C.~Wilson, eds., vol.~81 of Proceedings of
  Machine Learning Research, {PMLR}, 2018, pp.~77--91.

\bibitem{calders2010three}
{\sc T.~Calders and S.~Verwer}, {\em Three naive bayes approaches for
  discrimination-free classification}, Data Mining and Knowledge Discovery, 21
  (2010), pp.~277--292.

\bibitem{DBLP:conf/ijcai/CelisHV18}
{\sc L.~E. Celis, L.~Huang, and N.~K. Vishnoi}, {\em Multiwinner voting with
  fairness constraints}, in Proceedings of the Twenty-Seventh International
  Joint Conference on Artificial Intelligence, {IJCAI} 2018, July 13-19, 2018,
  Stockholm, Sweden, J.~Lang, ed., ijcai.org, 2018, pp.~144--151.

\bibitem{DBLP:conf/icalp/CelisSV18}
{\sc L.~E. Celis, D.~Straszak, and N.~K. Vishnoi}, {\em Ranking with fairness
  constraints}, in 45th International Colloquium on Automata, Languages, and
  Programming, {ICALP} 2018, July 9-13, 2018, Prague, Czech Republic,
  I.~Chatzigiannakis, C.~Kaklamanis, D.~Marx, and D.~Sannella, eds., vol.~107
  of LIPIcs, Schloss Dagstuhl - Leibniz-Zentrum f{\"{u}}r Informatik, 2018,
  pp.~28:1--28:15.

\bibitem{charlin2013toronto}
{\sc L.~Charlin and R.~Zemel}, {\em The toronto paper matching system: an
  automated paper-reviewer assignment system},  (2013).

\bibitem{chierichetti2017fair}
{\sc F.~Chierichetti, R.~Kumar, S.~Lattanzi, and S.~Vassilvitskii}, {\em Fair
  clustering through fairlets}, in Advances in Neural Information Processing
  Systems, 2017, pp.~5029--5037.

\bibitem{DBLP:conf/aistats/Chierichetti0LV19}
{\sc F.~Chierichetti, R.~Kumar, S.~Lattanzi, and S.~Vassilvitskii}, {\em
  Matroids, matchings, and fairness}, in The 22nd International Conference on
  Artificial Intelligence and Statistics, {AISTATS} 2019, 16-18 April 2019,
  Naha, Okinawa, Japan, K.~Chaudhuri and M.~Sugiyama, eds., vol.~89 of
  Proceedings of Machine Learning Research, {PMLR}, 2019, pp.~2212--2220.

\bibitem{chouldechova2017fair}
{\sc A.~Chouldechova}, {\em Fair prediction with disparate impact: A study of
  bias in recidivism prediction instruments}, Big data, 5 (2017), pp.~153--163.

\bibitem{DBLP:conf/kdd/Corbett-DaviesP17}
{\sc S.~Corbett{-}Davies, E.~Pierson, A.~Feller, S.~Goel, and A.~Huq}, {\em
  Algorithmic decision making and the cost of fairness}, in Proceedings of the
  23rd {ACM} {SIGKDD} International Conference on Knowledge Discovery and Data
  Mining, Halifax, NS, Canada, August 13 - 17, 2017, {ACM}, 2017, pp.~797--806.

\bibitem{crowson2016assessing}
{\sc C.~S. Crowson, E.~J. Atkinson, and T.~M. Therneau}, {\em Assessing
  calibration of prognostic risk scores}, Statistical methods in medical
  research, 25 (2016), pp.~1692--1706.

\bibitem{cygan2017hitting}
{\sc M.~Cygan, D.~Marx, M.~Pilipczuk, and M.~Pilipczuk}, {\em Hitting forbidden
  subgraphs in graphs of bounded treewidth}, Information and Computation, 256
  (2017), pp.~62--82.

\bibitem{datta2015automated}
{\sc A.~Datta, M.~C. Tschantz, and A.~Datta}, {\em Automated experiments on ad
  privacy settings: A tale of opacity, choice, and discrimination}, Proceedings
  on privacy enhancing technologies, 2015 (2015), pp.~92--112.

\bibitem{dwork2012fairness}
{\sc C.~Dwork, M.~Hardt, T.~Pitassi, O.~Reingold, and R.~Zemel}, {\em Fairness
  through awareness}, in Proceedings of the 3rd innovations in theoretical
  computer science conference, 2012, pp.~214--226.

\bibitem{dwork2018group}
{\sc C.~Dwork and C.~Ilvento}, {\em Group fairness under composition}, in
  Proceedings of the 2018 Conference on Fairness, Accountability, and
  Transparency (FAT* 2018), 2018.

\bibitem{ebadian2022optimized}
{\sc S.~Ebadian, A.~Kahng, D.~Peters, and N.~Shah}, {\em Optimized distortion
  and proportional fairness in voting}, in Proceedings of the 23rd ACM
  Conference on Economics and Computation, 2022, pp.~563--600.

\bibitem{feldman2015certifying}
{\sc M.~Feldman, S.~A. Friedler, J.~Moeller, C.~Scheidegger, and
  S.~Venkatasubramanian}, {\em Certifying and removing disparate impact}, in
  proceedings of the 21th ACM SIGKDD international conference on knowledge
  discovery and data mining, 2015, pp.~259--268.

\bibitem{freemantwo}
{\sc R.~Freeman, E.~Micha, and N.~Shah}, {\em Two-sided matching meets fair
  division},  (2020).

\bibitem{garb1997race}
{\sc H.~N. Garb}, {\em Race bias, social class bias, and gender bias in
  clinical judgment}, Clinical Psychology: Science and Practice, 4 (1997),
  pp.~99--120.

\bibitem{DBLP:journals/datamine/Garcia-SorianoB20}
{\sc D.~Garc{\'{\i}}a{-}Soriano and F.~Bonchi}, {\em Fair-by-design matching},
  Data Min. Knowl. Discov., 34 (2020), pp.~1291--1335.

\bibitem{GoelYF18}
{\sc N.~Goel, M.~Yaghini, and B.~Faltings}, {\em Non-discriminatory machine
  learning through convex fairness criteria}, in Proceedings of the 2018
  {AAAI/ACM} Conference on AI, Ethics, and Society, {AIES} 2018, New Orleans,
  LA, USA, February 02-03, 2018, J.~Furman, G.~E. Marchant, H.~Price, and
  F.~Rossi, eds., {ACM}, 2018, p.~116.

\bibitem{GuptaRSZ19}
{\sc S.~Gupta, S.~Roy, S.~Saurabh, and M.~Zehavi}, {\em Parameterized
  algorithms and kernels for rainbow matching}, Algorithmica, 81 (2019),
  pp.~1684--1698.

\bibitem{DBLP:journals/algorithmica/HuangKM016}
{\sc C.~Huang, T.~Kavitha, K.~Mehlhorn, and D.~Michail}, {\em Fair matchings
  and related problems}, Algorithmica, 74 (2016), pp.~1184--1203.

\bibitem{huang2019coresets}
{\sc L.~Huang, S.~Jiang, and N.~Vishnoi}, {\em Coresets for clustering with
  fairness constraints}, in Advances in Neural Information Processing Systems,
  2019, pp.~7589--7600.

\bibitem{impagliazzo2001complexity}
{\sc R.~Impagliazzo and R.~Paturi}, {\em On the complexity of k-sat}, Journal
  of Computer and System Sciences, 62 (2001), pp.~367--375.

\bibitem{DBLP:conf/nips/JosephKMR16}
{\sc M.~Joseph, M.~J. Kearns, J.~Morgenstern, and A.~Roth}, {\em Fairness in
  learning: Classic and contextual bandits}, in Advances in Neural Information
  Processing Systems 29: Annual Conference on Neural Information Processing
  Systems 2016, December 5-10, 2016, Barcelona, Spain, D.~D. Lee, M.~Sugiyama,
  U.~von Luxburg, I.~Guyon, and R.~Garnett, eds., 2016, pp.~325--333.

\bibitem{kamada2020fair}
{\sc Y.~Kamada and F.~Kojima}, {\em Fair matching under constraints: Theory and
  applications},  (2020).

\bibitem{DBLP:conf/icdm/KamishimaAS11}
{\sc T.~Kamishima, S.~Akaho, and J.~Sakuma}, {\em Fairness-aware learning
  through regularization approach}, in Data Mining Workshops (ICDMW), 2011
  {IEEE} 11th International Conference on, Vancouver, BC, Canada, December 11,
  2011, M.~Spiliopoulou, H.~Wang, D.~J. Cook, J.~Pei, W.~Wang, O.~R.
  Za{\"{\i}}ane, and X.~Wu, eds., {IEEE} Computer Society, 2011, pp.~643--650.

\bibitem{kesavan2022proportional}
{\sc D.~Kesavan, E.~Periyathambi, and A.~Chokkalingam}, {\em A proportional
  fair scheduling strategy using multiobjective gradient-based african buffalo
  optimization algorithm for effective resource allocation and interference
  minimization}, International Journal of Communication Systems, 35 (2022),
  p.~e5003.

\bibitem{klaus2006procedurally}
{\sc B.~Klaus and F.~Klijn}, {\em Procedurally fair and stable matching},
  Economic Theory, 27 (2006), pp.~431--447.

\bibitem{kleinberg2017inherent}
{\sc J.~Kleinberg, S.~Mullainathan, and M.~Raghavan}, {\em Inherent trade-offs
  in the fair determination of risk scores}, in 8th Innovations in Theoretical
  Computer Science Conference (ITCS 2017), Schloss Dagstuhl-Leibniz-Zentrum
  fuer Informatik, 2017.

\bibitem{kurata2017controlled}
{\sc R.~Kurata, N.~Hamada, A.~Iwasaki, and M.~Yokoo}, {\em Controlled school
  choice with soft bounds and overlapping types}, Journal of Artificial
  Intelligence Research, 58 (2017), pp.~153--184.

\bibitem{lu2022optimization}
{\sc Y.~Lu}, {\em The optimization of automated container terminal scheduling
  based on proportional fair priority}, Mathematical Problems in Engineering,
  2022 (2022).

\bibitem{NaorSS95}
{\sc M.~Naor, L.~J. Schulman, and A.~Srinivasan}, {\em Splitters and
  near-optimal derandomization}, in Proceedings of the 36th Annual Symposium on
  Foundations of Computer Science (FOCS 1995), IEEE, 1995, pp.~182--191.

\bibitem{nguyen2022nash}
{\sc M.~H. Nguyen, M.~Baiou, V.~H. Nguyen, and T.~Q.~T. Vo}, {\em Nash fairness
  solutions for balanced tsp}, in 10th International Network Optimization
  Conference (INOC), 2022.

\bibitem{ristoski2018machine}
{\sc P.~Ristoski, P.~Petrovski, P.~Mika, and H.~Paulheim}, {\em A machine
  learning approach for product matching and categorization}, Semantic web, 9
  (2018), pp.~707--728.

\bibitem{st2022adaptation}
{\sc W.~St-Arnaud, M.~Carvalho, and G.~Farnadi}, {\em Adaptation, comparison
  and practical implementation of fairness schemes in kidney exchange
  programs}, arXiv preprint arXiv:2207.00241,  (2022).

\bibitem{DBLP:conf/mfcs/Stamoulis14}
{\sc G.~Stamoulis}, {\em Approximation algorithms for bounded color matchings
  via convex decompositions}, in Mathematical Foundations of Computer Science
  2014 - 39th International Symposium, {MFCS} 2014, Budapest, Hungary, August
  25-29, 2014. Proceedings, Part {II}, E.~Csuhaj{-}Varj{\'{u}},
  M.~Dietzfelbinger, and Z.~{\'{E}}sik, eds., vol.~8635 of Lecture Notes in
  Computer Science, Springer, 2014, pp.~625--636.

\bibitem{DBLP:conf/kdd/ThanhRT11}
{\sc B.~L. Thanh, S.~Ruggieri, and F.~Turini}, {\em k-nn as an implementation
  of situation testing for discrimination discovery and prevention}, in
  Proceedings of the 17th {ACM} {SIGKDD} International Conference on Knowledge
  Discovery and Data Mining, San Diego, CA, USA, August 21-24, 2011,
  C.~Apt{\'{e}}, J.~Ghosh, and P.~Smyth, eds., {ACM}, 2011, pp.~502--510.

\bibitem{Yannakakis78}
{\sc M.~Yannakakis}, {\em Node- and edge-deletion np-complete problems}, in
  Proceedings of the 10th Annual {ACM} Symposium on Theory of Computing, May
  1-3, 1978, San Diego, California, {USA}, R.~J. Lipton, W.~A. Burkhard, W.~J.
  Savitch, E.~P. Friedman, and A.~V. Aho, eds., {ACM}, 1978, pp.~253--264.

\bibitem{zhang2022equality}
{\sc G.~Zhang, S.~Malekmohammadi, X.~Chen, and Y.~Yu}, {\em Equality is not
  equity: Proportional fairness in federated learning}, arXiv preprint
  arXiv:2202.01666,  (2022).

\end{thebibliography}
